\documentclass[conference]{IEEEtran}
\ifCLASSINFOpdf
\else
\fi
\usepackage{amsmath,amssymb,amsthm}

\usepackage{graphicx}
\newtheorem{theorem}{Theorem}
\usepackage{algorithm}
\usepackage{algorithmic}
\usepackage{amsmath}
\usepackage{cite}
\usepackage{caption}

\begin{document}
\title{Optimal Placement of a UAV to Maximize the
	Lifetime of Wireless Devices}

\author{\IEEEauthorblockN{Hazim Shakhatreh}
	\IEEEauthorblockA{Department of
		Electrical and\\Computer Engineering\\
		New Jersey Institute of Technology\\
		Email: hms35@njit.edu}
	\and
	\IEEEauthorblockN{Abdallah Khreishah}
	\IEEEauthorblockA{Department of
		Electrical and\\Computer Engineering\\
		New Jersey Institute of Technology\\
		Email: abdallah@njit.edu}}

\maketitle

\begin{abstract}
Unmanned aerial vehicles (UAVs) can be used as
aerial wireless base stations when cellular networks go down. Prior studies on UAV-based wireless coverage typically consider downlink scenarios from an aerial base station to ground users. In this paper, we consider an uplink scenario under disaster situations (such as earthquakes or floods), when cellular networks are down. We formulate the problem of optimal UAV placement, where the objective is to determine the placement of a single UAV such that the sum of time durations of uplink transmissions is maximized. We prove that the constraint sets of problem can be represented by the intersection of half spheres and the region formed by this intersection is a convex set in terms of two variables. This proof enables us to transform our problem to an optimization problem with two variables. We also prove that the objective function of the transformed problem is a concave function under a restriction on the minimum altitude of the UAV and propose a gradient projection-based algorithm to find the optimal location of the UAV. We validate the analysis by simulations and demonstrate the effectiveness of the proposed algorithm under different cases. 
\end{abstract}

\IEEEpeerreviewmaketitle

\section{Introduction}
UAVs can be used to provide wireless coverage during
emergency cases where each UAV serves as an aerial wireless
base station when the cellular network goes down~\cite{shakhatreh2016continuous}. They
can also be used to supplement the ground base station in
order to provide better coverage and higher data rates for the
users~\cite{bor2016efficient}.

Prior studies on UAV-based wireless coverage typically consider downlink scenarios from a UAV to ground users. The authors in~\cite{mozaffari2015drone} investigate the downlink coverage performance of a UAV, where the objective is to find the optimal UAV altitude which leads to the maximum ground coverage and the minimum transmit power. In~\cite{mozaffari2016optimal}, the authors consider the downlink scenario, where the goal is to minimize the total required transmit power of UAVs while satisfying the users’ rate requirements. In~\cite{shakhatreh2017providing,shakhatreh2017efficient,shakhatreh2017maximizing,sawalmeh2017providing}, the authors propose using a UAV to provide wireless coverage for users during emergency cases and special events. Due to the limited transmit power of the UAV, the authors in~\cite{shakhatreh2017indoor} study the problem of minimizing the number of UAVs required to cover the indoor users.

Only few studies consider the uplink scenario in which the ground wireless devices transmit data to a UAV. The authors in~\cite{zeng2016throughput} study the throughput maximization problem in UAV relaying systems by optimizing the source/relay transmit power along with the UAV trajectory, subject to practical mobility constraints. In~\cite{yang2017energy}, the authors present a UAV enabled data collection system, where a UAV is dispatched to collect a given amount of data from ground
terminals at fixed location. They aim to find the optimal ground terminal transmit power and UAV trajectory that achieve
different Pareto optimal energy trade-offs between the ground terminal and the UAV.

Under disaster situations (such as earthquakes or floods), users may not be able to communicate with remote-undamaged terrestrial ground stations due to the limited transmit power of wireless devices. They are also not able to recharge their wireless devices due to physical damage to energy infrastructure. In the case of Hurricane Katrina, about 700,000 customers in Louisiana and almost 200,000 in Mississippi lost power~\cite{kwasinski2006hurricane}. In such situations, providing wireless coverage becomes more important, since people in the disaster area seek information for themselves to learn about the emergency event, locate their family and friends, and report their safety~\cite{taniguchi2012effect,simon2015socializing}. In this paper, we are motivated to explore if the placement of UAV can enhance the time durations of uplink transmissions of wireless devices. To the best of our knowledge, this is the first work that proposes using a UAV to maximize the sum of time durations of uplink transmissions under disaster situations, where the ground users are not able to recharge their wireless devices due to physical damage to energy infrastructure. We summarize our main contributions as follows. First, we formulate the problem of optimal UAV placement, where the objective is to maximize the lifetime of wireless devices, the lifetime represents the sum of time durations of uplink transmissions. Second, we prove that the constraint sets of problem can be represented by the intersection of half spheres and the region formed by this intersection is a convex set in terms of two variables. This proof enables us to transform our problem into an optimization problem with two variables. Third, we prove that the objective function of the transformed problem is a concave function under a restriction on the minimum altitude of UAV and propose a gradient projection-based algorithm to find the optimal location of UAV.

The rest of this paper is organized as follows. In Section II,
we describe the system model. In Section III, we formulate the problem of UAV placement with an objective of maximizing the sum of time durations of uplink transmissions and present a gradient projection-based algorithm to find the optimal location of the UAV under a restriction on the minimum altitude. Finally, we present our numerical results in Section IV and make concluding remarks in Section V.

\section{SYSTEM MODEL}

Let $(X_u, Y_u, Z_u)$ denote the 3D location of the UAV. We assume that $|I|$ ground users are located according to a probability distribution $f(x, y)$. Each user $i\in I$ has a wireless device with residual energy $E_i$ and the maximum transmit power of each device is $P_{max}$. The users  must be served by a single UAV that acts as aerial base station to collect data from users as shown in Figure~\ref{fig1}. We consider an uplink scenario in which the ground users adopt a frequency division multiple access (FDMA) technique to transmit data to a UAV at a desired data rate $R$. FDMA allocates one subchannel to each user for communications and hence the channels do not interfere with one another. We also assume that each user $i\in I$ is served by a UAV for a time $\tau_i$ seconds and this time depends on the residual energy of wireless device represented by the battery level $E_i$. The time duration of uplink transmission $\tau_i$ must be greater than or equal to $\tau_{th}$.
 
In this paper, we assume that the wireless channel between ground user $i$ and UAV  is line of sight dominated,
so that the free space path loss model is adopted similar to~\cite{zeng2016throughput}~and \cite{yang2017energy}. The path loss is given as follows:
\begin{equation}
\begin{split}
L_i=\left(\dfrac{4\pi d_i f}{c}\right)^2
\end{split}
\end{equation}
where $d_i=\sqrt{(X_u-x_i)^2+(Y_u-y_i)^2+(Z_u)^2}$ is the distance between ground user $i$ and a UAV, $f$ is frequency (in Hz) and $c$ is the speed of light (in m/s). Notice that when the distance between a ground user and UAV (i.e., $d_i$) increases, the transmit power (i.e., $p_i$) increases, while the time duration of uplink transmission (i.e., $\tau_i$) decreases.

\begin{figure}[!t]
	\centering
	\includegraphics[scale=0.25]{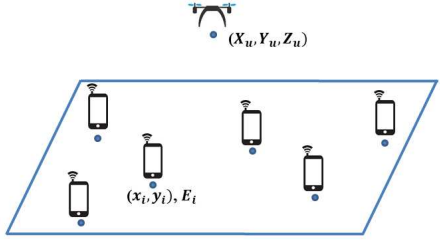}
	\caption{Ground users transmitting data to a UAV}
	\label{fig1}
\end{figure}
\begin{figure*}[t]
		\hspace{1.3cm}
	\begin{minipage}[b]{0.35\linewidth}
		\centering
		\includegraphics[width=\textwidth]{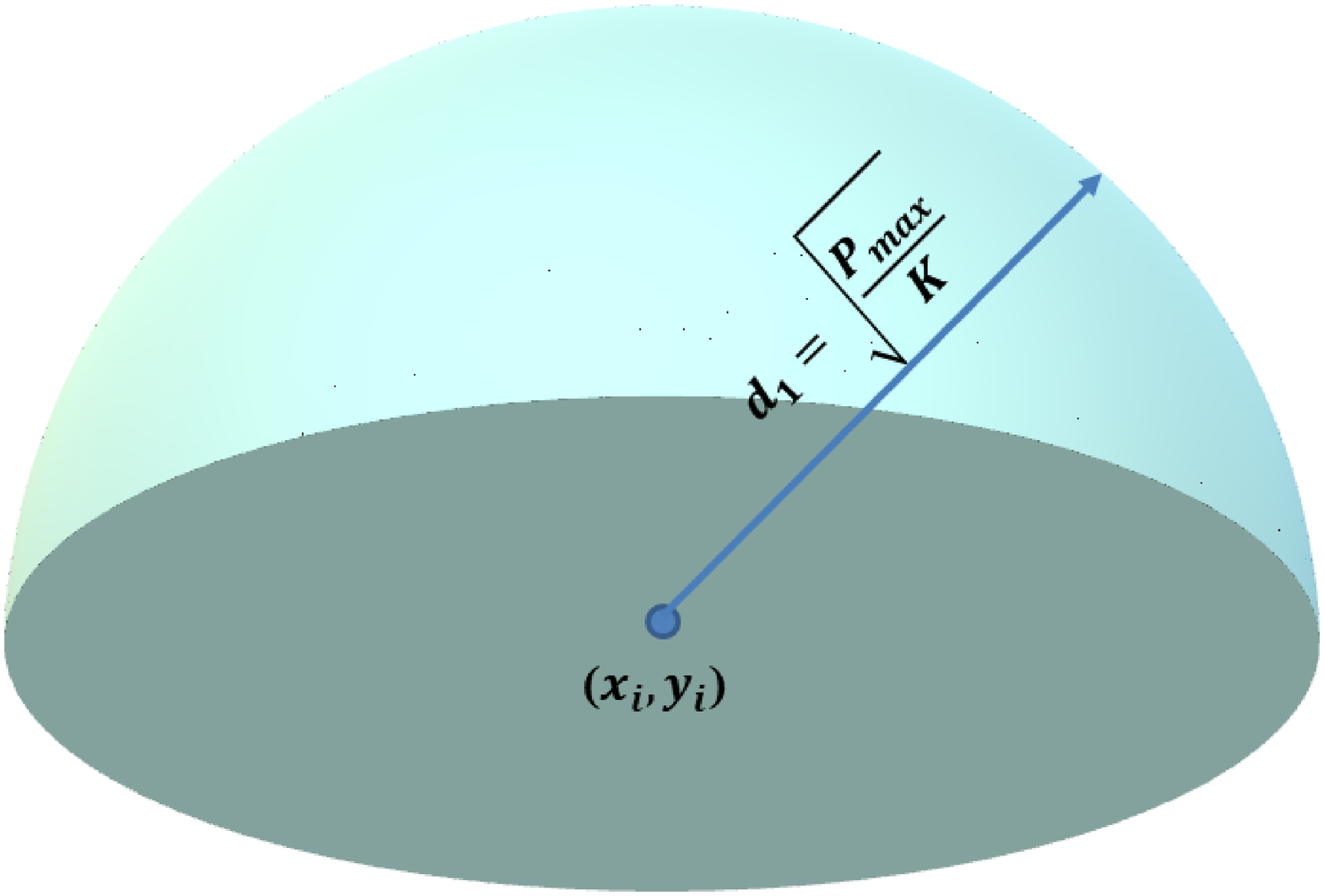}
		\caption{The range of distances that satisfies constraint set (4.a).
		}
		\label{fig2}
	\end{minipage}
	\hspace{2.7cm}
	\begin{minipage}[b]{0.35\linewidth}
		\centering
		\includegraphics[width=\textwidth]{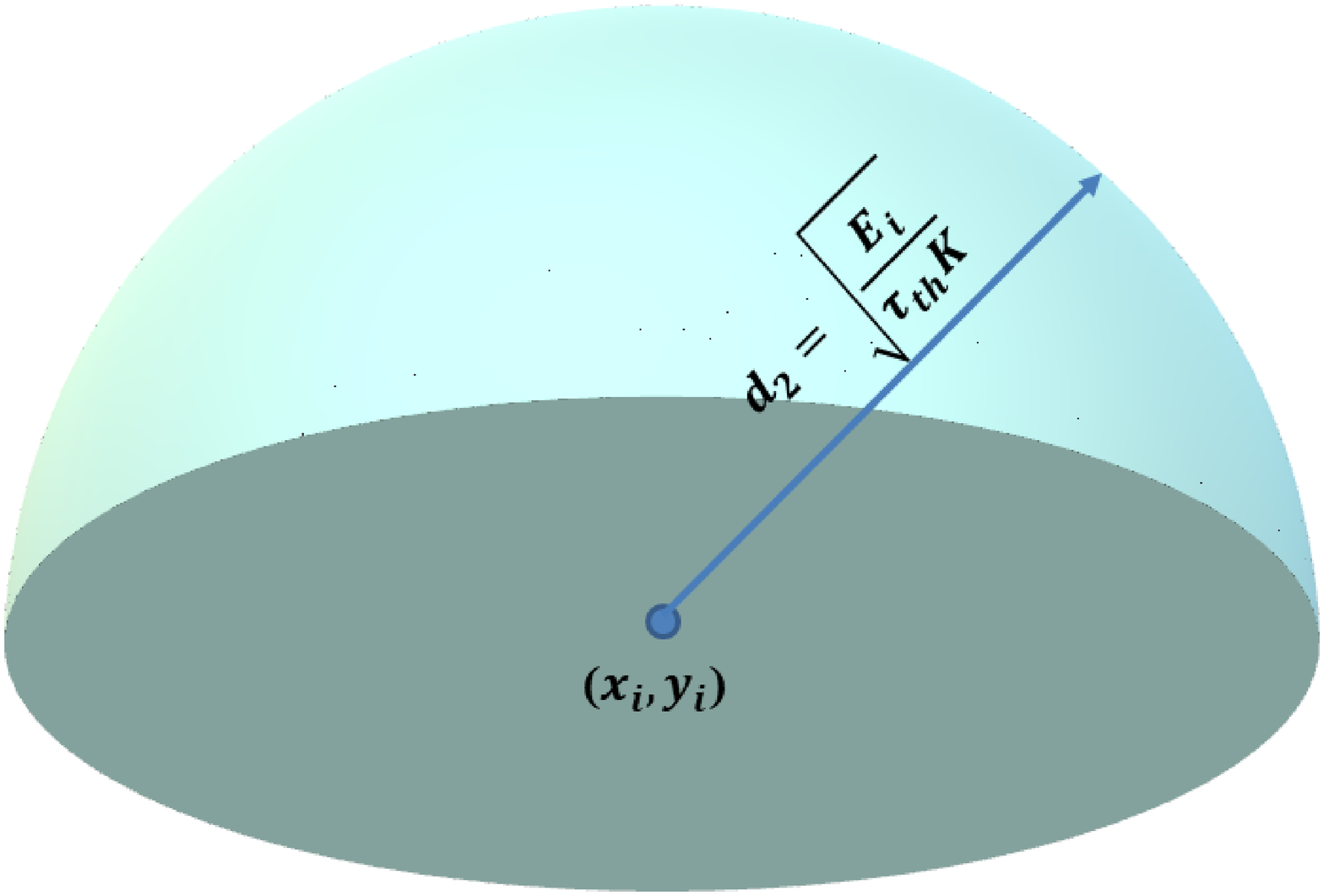}
		\caption{The range of distances that satisfies constraint sets (4.b) and (4.c).
		}
		\label{fig3}
	\end{minipage}
	\vspace{10pt}
\end{figure*}

\section{PROBLEM FORMULATION}
Consider a transmission between a user located at $(x_i, y_i)$ and a UAV located at $(X_u, Y_u, Z_u)$. The rate for user $i$ is given by:
\begin{equation}
\begin{split}
C_{i}=B_ilog_{2}\left(1+\dfrac{p_{i}/L_i}{N}\right)
\end{split}
\end{equation}
where $B_i$ is the transmission bandwidth of user $i$, $p_{i}$ is the transmit power from user $i$ to the UAV, $L_i$ is the path loss between user $i$ and the UAV and $N$ is the noise power. 

Let us assume that all users have the same data rate $R$ and each user has a channel with bandwidth equals $B/|I|$, where $B$ is the UAV bandwidth and $|I|$ is the number of ground users. Then, the minimum power required to satisfy this rate for each user is given by:
\begin{equation}
\begin{split}
p_{i}=\left(2^{\frac{R.|I|}{B}}-1\right).N.L_i
\end{split}
\end{equation}
Our goal is to find the optimal location of the UAV such that the lifetime of wireless devices defined by $T$ is maximized. Here, the lifetime $T$ represents the sum of time durations of uplink transmissions. Our problem can be formulated as:
\begin{equation}
\begin{split}
\max_{(X_{u}, Y_{u}, Z_{u}),\tau_i} T=\sum_{i=1}^{\arrowvert I \arrowvert}\tau_i~~~~~~~~~~~~~~~~~~~~~~~~~~~~~~~~~\\
subject ~to~~~~~~~~~~~~~~~~~~~~~~~~~~~~~~~~~~~~~~~~~~~~~~~~~~~~~~~~~~~\\
\left(2^{\frac{R.|I|}{B}}-1\right).N.L_i \leq P_{max}~~~~~~~~~~~~~~\forall i \in I~~~~~(4.a)\\
\tau_i \geq \tau_{th}~~~~~~~~~~~~~~~~~\forall i \in I~~~~~(4.b)\\
\tau_i.\left(2^{\frac{R.|I|}{B}}-1\right).N.L_i \leq E_{i}~~~~~~~~~~~~~~~~~\forall i \in I~~~~~(4.c)\\
x_{min}\leq X_{u}\leq x_{max}~~~~~~~~~~~~~~~~~~~~~~(4.d)\\
y_{min}\leq Y_{u}\leq y_{max}~~~~~~~~~~~~~~~~~~~~~~(4.e)\\
z_{min}\leq Z_{u}\leq z_{max}~~~~~~~~~~~~~~~~~~~~~~(4.f)\\
\end{split}
\end{equation}
Here, constraint set (4.a) ensures that the transmit power of each wireless device should not exceed its maximum transmit power $P_{max}$. Constraint set (4.b) guarantees that each ground user $i\in I$ is served by UAV for a time greater than $\tau_{th}$ seconds. Constraint set (4.c) ensures
that the total energy consumed by user's device should not exceed
its battery energy level $E_i$. Constraints (4.d-4.f) represent the minimum and maximum allowed values for $X_{u}$, $Y_{u}$ and $Z_{u}$.

From (1), we can notice that the optimal altitude of UAV that maximizes the lifetime of wireless devices is equal to $z_{min}$, which could correspond to the minimum altitude due to safety consideration~\cite{yang2017energy}. Now, our objective is to find the optimal 2D placement of the UAV such that the lifetime of wireless devices is maximized. Even though the problem has a number of nonlinear constraints, we can transform (4) to an optimization problem with two variables by proving that the constraint sets (4.a-4.c) can be represented by the intersection of half spheres and the region formed by this intersection is a convex set in terms of $(X_u,Y_u)$.
 \begin{theorem}
 	The constraint sets (4.a-4.c) can be represented by the intersection of half spheres and the region formed by this intersection is a convex set in terms of $(X_u,Y_u)$. 
 \end{theorem}
 \begin{proof}
 	From (1) and (3), the transmit power of ground user $i$ is given by:
 	\begin{equation}
 	\begin{split}
 	p_i=\left(2^{\frac{R.|I|}{B}}-1\right).N.\left(\dfrac{4\pi d_i f}{c}\right)^2=Kd^2_i
 	\end{split}
 	\end{equation}
 	where $K$ is a constant and equals $\left(2^{\frac{R.|I|}{B}}-1\right). N.\left(\dfrac{4\pi f}{c}\right)^2$. Now, to satisfy constraint set (4.a), $p_i$ must be less than $P_{max}$. From (5), the range of distances $d_1$ that satisfies the constraint set (4.a) is given by:
 	 	\begin{equation}
 	 	d_1\leq\sqrt{\dfrac{P_{max}}{K}}
 	 	 	\end{equation}
 	 The range of distances $d_1$ represents a half sphere with radius $\sqrt{\dfrac{P_{max}}{K}}$ as shown in Figure 2. To satisfy constraint sets (4.b) and (4.c), $p_i$ must be less than $\dfrac{E_i}{\tau_{th}}$. From (5), the range of distances $d_2$ that satisfies constraint sets (4.b) and (4.c) is given by:
 	 \begin{equation}
 	 d_2\leq\sqrt{\dfrac{E_i}{\tau_{th}K}}
 	 \end{equation}

  The range of distances $d_2$ also represents a half sphere with radius $\sqrt{\dfrac{E_i}{\tau_{th}K}}$ as shown in Figure 3. For each ground user $i$, the transmit power $p_i$ and the time duration of uplink transmission $\tau_{i}$ that satisfy the constraint sets (4.a)-(4.c) can be represented by a half sphere with radius:
   	 \begin{equation}
   min\left\{\sqrt{\dfrac{P_{max}}{K}},\sqrt{\dfrac{E_i}{\tau_{th}K}}\right\}
 	 \end{equation}
 	 The half sphere is a convex set and the intersection of convex sets is also a convex~\cite{boyd2004convex}.
 \end{proof}
 From Theorem 1, we can represent the transmit power $p_i$ and the time duration of uplink transmission $\tau_{i}$ constraints for each user as a half sphere. The intersection of all half spheres represents the convex set $V$ that satisfies the constraint sets (4.a)-(4.c) for all users. Therefore, we restrict the placement of UAV to be in $V$. Now, we can transform our problem to an optimization problem with two variables $(X_u, Y_u)$, where $(X_u, Y_u)$ represent the 2D placement of UAV. The proposed algorithm to find the convex set $V$ is shown in Algorithm 1 as follows: The inputs are the locations of users, the maximum transmit power, the energy of each wireless device, the data rate, the total bandwidth, the operating frequency, the noise power and the threshold time duration of uplink transmission. In steps (10-12), we find the range of distances that satisfies the maximum power constraint. On the other hand, steps (13-16) find the range of distances that satisfies the threshold time duration of uplink transmission. In steps (17-21), we find the convex set $V$.

\begin{algorithm}
	\begin{algorithmic}
		\STATE 1: \textbf{Input:}
		\STATE 2: The locations of $|I|$ ground users.
		\STATE 3: The maximum transmit power of wireless device $P_{max}$.
		\STATE 4: The energy of each wireless device ${E_i}$.
		\STATE 5: The data rate $R$.
		\STATE 6: The total bandwidth $B$.
		\STATE 7: The noise $N$.
		\STATE 8: The operating frequency $f$.
		\STATE 9: The threshold time duration of uplink transmission. $\tau_{th}$.
		\STATE 10: Find $K=(2^{\frac{R.|I|}{B}}-1).N.(\frac{4\pi f}{c})^2$
		\STATE 11: For each user, the range of distances that satisfies the maximum power constraint is given by:
		\STATE 12: $d_1\leq\sqrt{\frac{P_{max}}{K}}$
		\STATE 13: The range of distances that satisfies the threshold time duration of uplink transmission is given by:
		\STATE 14: \textbf{For} $i=1:|I|$
		\STATE 15:~~~~~$d_{2,i}\leq\sqrt{\frac{E_i}{\tau_{th}K}}$
		\STATE 16:\textbf{End}
		\STATE 17: For each user, the range of distances that satisfies the problem constraints is given by:
		\STATE 18: \textbf{For} $i=1:|I|$
		\STATE 19:~~~~~$d_i=min\{d_1,d_{2,i}\}$
		\STATE 20:\textbf{End}
		\STATE 21:The convex set $V$ is given by:
		\STATE $\bigcap\limits_{i=1}^{|I|} \{(X_u,Y_u)\in \textbf{R}^2|\sqrt{(X_u-x_i)^2+(Y_u-y_i)^2+z^2_{min}}\leq d_i\}$
	\end{algorithmic}
	\caption{The Convex Set Algorithm}
\end{algorithm}

Now, the objective function in (4) can be represented as:
\begin{equation}
\begin{split}
\sum_{i=1}^{\arrowvert I \arrowvert}\tau_i=\sum_{i=1}^{\arrowvert I \arrowvert}\dfrac{E_i}{p_i}=\sum_{i=1}^{\arrowvert I \arrowvert}\dfrac{E_i}{(2^{\frac{R.|I|}{B}}-1).N.L_i}=~~~~~~~~~~\\
\sum_{i=1}^{\arrowvert I \arrowvert}\dfrac{E_i}{(2^{\frac{R.|I|}{B}}-1).N.(\dfrac{4\pi d_i f}{c})^2}=\sum_{i=1}^{\arrowvert I \arrowvert}\dfrac{E_i}{Kd^2_i}~~~~~~~~~~~~~~
\end{split}
\end{equation}
Since $K$ is constant, our problem can be formulated as:
\begin{equation}
\begin{split}
\max_{(X_{u}, Y_{u})} \sum_{i=1}^{\arrowvert I \arrowvert}\dfrac{E_i}{(X_u-x_i)^2+(Y_u-y_i)^2+z^2_{min}}~~~~~~~~~~~~~~~~~~~~~~~~~~~~~~~~~\\
subject ~to~~~~~~~~~~~~~~~~~~~~~~~~~~~~~~~~~~~~~~~~~~~~~~~~~~~~~~~~~~~~~~~~~~~~~~~~~~\\
(X_{u}, Y_{u}) \in~~~~~~~~~~~~~~~~~~~~~~~~~~~~~~~~~~~~~~~~~~~~~~~~~~~~~~~~~~~~~~~~~~~~~~~~~\\
 \bigcap\limits_{i=1}^{|I|} \{(X_u,Y_u)\in \textbf{R}^2|\sqrt{(X_u-x_i)^2+(Y_u-y_i)^2+z^2_{min}}\leq d_i\}~~~~~~~~~~~~~~~\\
\end{split}
\end{equation}
The transformed problem (10) is equivalent to problem (4). In the next theorem, we prove that the objective function is concave under a restriction on the minimum altitude of UAV $z_{min}$. This theorem enables us to find the optimal $(X_u, Y_u)$ placement for UAV. 
\begin{theorem}
	The objective function of (10) is concave if the minimum altitude of UAV $z_{min}$ is greater than $\sqrt{3}d_{max}$.
\end{theorem}
\begin{proof}
	We know that the nonnegative weighted sums preserve the concavity of function~\cite{boyd2004convex}. Since $E_i>0,\forall i\in I$, we need to prove that (11) is a concave function. 
	\begin{equation}
	\begin{split}
	f=\dfrac{1}{(X_u-x_i)^2+(Y_u-y_i)^2+z^2_{min}}, \forall i\in I
	\end{split}
	\end{equation}
	Using the second order condition, the function $f$ is concave if the Hessian is negative semidefinite~\cite{boyd2004convex}. Now, the Hessian is negative semidefinite if we satisfy these conditions:
		\begin{equation}
		\begin{split}
		(a) \dfrac{d^2f}{dX^2_u}\leq 0,~~~~~\forall i\in I~~~~~~~~~~~~~~~~~~~~~~~~~~~~~~~~~~~~\\
		(b) \dfrac{d^2f}{dY^2_u}\leq 0,~~~~~\forall i\in I~~~~~~~~~~~~~~~~~~~~~~~~~~~~~~~~~~~~\\
		(c) \dfrac{d^2f}{dX^2_u}\dfrac{d^2f}{dY^2_u}-(\dfrac{d^2f}{dXdY_u})^2\geq 0,~~~~~\forall i\in I~~~~~~~~~~~~~\\
		\end{split}
		\end{equation}
		To check the first condition, we need to find $\dfrac{d^2f}{dX^2_u}$:
		\begin{equation}
		\begin{split}
		\dfrac{df}{dX_u}=\dfrac{-2(X_u-x_i)}{((X_u-x_i)^2+(Y_u-y_i)^2+z^2_{min})^2}~~~~~~~~~~~~~\\
		\dfrac{d^2f}{dX^2_u}=\dfrac{-2((X_u-x_i)^2+(Y_u-y_i)^2+z^2_{min})^2}{((X_u-x_i)^2+(Y_u-y_i)^2+z^2_{min})^4}+~~~~~~\\
		\dfrac{8(X_u-x_i)^2((X_u-x_i)^2+(Y_u-y_i)^2+z^2_{min})}{((X_u-x_i)^2+(Y_u-y_i)^2+z^2_{min})^4}\\
		=\dfrac{-2((X_u-x_i)^2+(Y_u-y_i)^2+z^2_{min})+8(X_u-x_i)^2}{((X_u-x_i)^2+(Y_u-y_i)^2+z^2_{min})^3}\\
		=\dfrac{6(X_u-x_i)^2-2(Y_u-y_i)^2-2z^2_{min}}{((X_u-x_i)^2+(Y_u-y_i)^2+z^2_{min})^3}~~~~~~~~~~~~~~~~~~~~
		\end{split}
		\end{equation}
		From (13), $\dfrac{d^2f}{dX^2_u}\leq 0,\forall i\in I$ if:
		\begin{equation}
		\begin{split}
		z^2_{min}>3(X_u-x_i)^2-(Y_u-y_i)^2,\forall i\in I~~~~~~~~~~~~~
	\end{split}
	\end{equation}
	Similarly, $\dfrac{d^2f}{dY^2_u}\leq 0,\forall i\in I$ if:
	\begin{equation}
	\begin{split}
	z^2_{min}>3(Y_u-y_i)^2-(X_u-x_i)^2,\forall i\in I~~~~~~~~~~~~~
	\end{split}
	\end{equation}
	To check the third condition, we need to find $\dfrac{d^2f}{dX^2_u}\dfrac{d^2f}{dY^2_u}-(\dfrac{d^2f}{dXdY_u})^2$:
	\begin{equation}
	\begin{split}
	\dfrac{d^2f}{dX_udY_u}=\dfrac{8(X_u-x_i)(Y_u-y_i)}{((X_u-x_i)^2+(Y_u-y_i)^2+z^2_{min})^3}
	\end{split}
	\end{equation}
	From (16), we get:
	\begin{equation*}
	\begin{split}
	\dfrac{d^2f}{dX^2_u}\dfrac{d^2f}{dY^2_u}-(\dfrac{d^2f}{dXdY_u})^2=~~~~~~~~~~~~~~~~~~~~~~~~~~~~~~~~~~~\\
	(\dfrac{-2((X_u-x_i)^2+(Y_u-y_i)^2+z^2_{min})+8(X_u-x_i)^2}{((X_u-x_i)^2+(Y_u-y_i)^2+z^2_{min})^3}\times\\
	\dfrac{-2((X_u-x_i)^2+(Y_u-y_i)^2+z^2_{min})+8(Y_u-y_i)^2}{((X_u-x_i)^2+(Y_u-y_i)^2+z^2_{min})^3})- \\
	\dfrac{64(X_u-x_i)^2(Y_u-y_i)^2}{((X_u-x_i)^2+(Y_u-y_i)^2+z^2_{min})^6}=~~~~~~~~~~~~~~~~~~~\\
	\dfrac{4((X_u-x_i)^2+(Y_u-y_i)^2+z^2_{min})^2}{((X_u-x_i)^2+(Y_u-y_i)^2+z^2_{min})^6}-~~~~~~~~~~~~~~~~~\\
	\dfrac{16(Y_u-y_i)^2((X_u-x_i)^2+(Y_u-y_i)^2+z^2_{min})}{((X_u-x_i)^2+(Y_u-y_i)^2+z^2_{min})^6}-~~~~\\
	\dfrac{16(X_u-x_i)^2((X_u-x_i)^2+(Y_u-y_i)^2+z^2_{min})}{((X_u-x_i)^2+(Y_u-y_i)^2+z^2_{min})^6}+~~~\\
		\end{split}
		\end{equation*}
		\begin{equation}
		\begin{split}
	\dfrac{64(X_u-x_i)^2(Y_u-y_i)^2}{((X_u-x_i)^2+(Y_u-y_i)^2+z^2_{min})^6}-~~~~~~~~~~~~~~~~~~\\
	\dfrac{64(X_u-x_i)^2(Y_u-y_i)^2}{((X_u-x_i)^2+(Y_u-y_i)^2+z^2_{min})^6}=~~~~~~~~~~~~~~~~~~\\
	\dfrac{4((X_u-x_i)^2+(Y_u-y_i)^2+z^2_{min})}{((X_u-x_i)^2+(Y_u-y_i)^2+z^2_{min})^5}+~~~~~~~~~~~~~~~~~~\\
	\dfrac{-16(X_u-x_i)^2-16(Y_u-y_i)^2}{((X_u-x_i)^2+(Y_u-y_i)^2+z^2_{min})^5}=~~~~~~~~~~~~~~~~~~\\
	\dfrac{-12(X_u-x_i)^2-12(Y_u-y_i)^2+4z^2_{min}}{((X_u-x_i)^2+(Y_u-y_i)^2+z^2_{min})^5}~~~~~~~~~~~~~~~~\\
	\end{split}
	\end{equation}
	From (17), $\dfrac{d^2f}{dX^2_u}\dfrac{d^2f}{dY^2_u}-(\dfrac{d^2f}{dXdY_u})^2\geq 0,\forall i\in I$ if:
	\begin{equation}
	\begin{split}
	z^2_{min}>3(X_u-x_i)^2+3(Y_u-y_i)^2,\forall i\in I~~~~~~~~~~~~~
	\end{split}
	\end{equation}
	From (14), (15) and (18), the Hessian is negative semidefinite if we satisfy these conditions:
	\begin{equation}
	\begin{split}
	(a) z^2_{min}>3(X_u-x_i)^2-(Y_u-y_i)^2,\forall i\in I~~~~~~~~~~~~~~~~\\
	(b) z^2_{min}>3(Y_u-y_i)^2-(X_u-x_i)^2,\forall i\in I~~~~~~~~~~~~~~~~\\
	(c) z^2_{min}>3(X_u-x_i)^2+3(Y_u-y_i)^2,\forall i\in I~~~~~~~~~~~~~~~
		\end{split}
		\end{equation}
		\begin{figure*}[!t]
			\centering
			\includegraphics[scale=0.28]{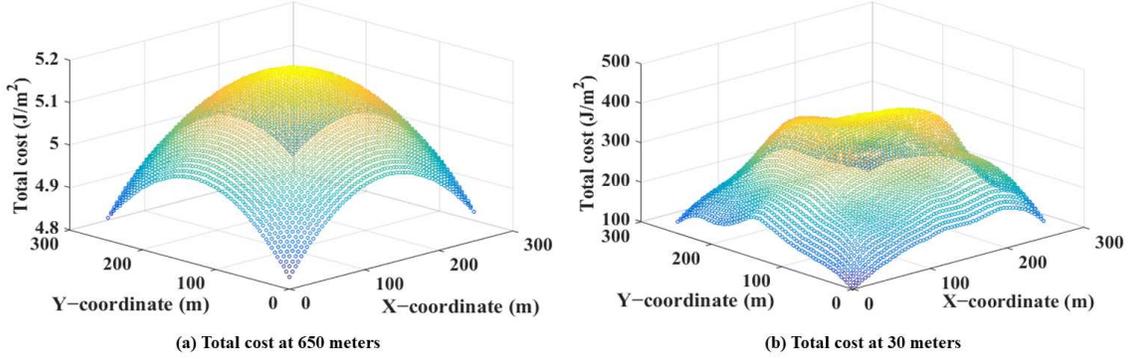}
			\caption{Total cost at different altitudes}
			\label{fig4}
		\end{figure*}
    From the three conditions in (19), we can notice that if we satisfy condition (c), we then satisfy conditions (a) and (b). Let us define $d_{max}$ as a maximum possible 2D distance in the geographical area (i.e., if the users are distributed in a circular geographical area, then $d_{max}$ is equal to the diameter of circle). From condition (c), if $z_{min}>\sqrt{3}d_{max}$ then the objective function of (10) is concave where $d_{max}>\sqrt{(X_u-x_i)^2+(Y_u-y_i)^2},\forall i\in I$. Here, we can notice that the altitude of UAV $z_{min}$ controls the concavity of the objective function.
    		\end{proof}
    		 Theorem 2 enables us to find the optimal $(X_u, Y_u)$ placement for UAV, when the altitude of UAV $z_{min}$ is greater than $\sqrt{3}d_{max}$, which is a practical altitude especially for high altitude platforms. In order to find this point, we propose the Gradient Projection Algorithm~\cite{bertsekas1989parallel}. The gradient algorithm is not applicable to constrained optimization problems, because even if we start inside the feasible
    		 region, an update can take us outside that region. A simple way to solve this problem is to project
    		 back to the feasible region whenever such a situation arises. The gradient projection algorithm
    		 is given by:
    		 \begin{equation}
    		 \begin{split}
    		 (X_u, Y_u)^{n+1}=[(X_u, Y_u)^n+\gamma.\triangledown F((X_u, Y_u)^n))]^+.
    		 \end{split}
    		 \end{equation}
    		 Here, $n$ is the iteration number, $\gamma$ is a positive step size, $\triangledown F$ is the gradient of the objective function in (10) and $[q]^+$ denotes the orthogonal projection of vector $q$ onto convex set $Q$. In particular, $[q]^+$ is defined by: .
    		  \begin{equation}
    		  [q]^+=arg~min_{w \in Q}||w-q||_2
    		  \end{equation}
    		  To find the orthogonal projection, we use the matlab function fmincon. The pseudo code of the gradient projection algorithm is shown in Algorithm 2.
    		  \begin{algorithm}
    		  	\begin{algorithmic}
    		  		\STATE \textbf{Input:}
    		  		\STATE The step tolerance $\epsilon$.
    		  		\STATE The step size $\gamma$.
    		  		\STATE The maximum number of iterations $n_{max}$.
    		  		\STATE \textbf{Initialize} $(X_u, Y_u)$
    		  		\STATE \textbf{For} $n$=1,2,..., $n_{max}$
    		  		\STATE      ~~~		 $(X_u, Y_u)^{n+1}=[(X_u, Y_u)^n+\gamma.\triangledown F((X_u, Y_u)^n))]^+$
    		  		\STATE ~~~~~~~~~~~\textbf{If} $\lVert$ $(X_u, Y_u)^n$ $-$ $(X_u, Y_u)^{n+1}$ $\rVert$ $<$ $\epsilon$
    		  		\STATE ~~~~~\textbf{Return:} \textbf{$(X_u, Y_u)_{opt}$} $=$ $(X_u, Y_u)^{n+1}$
    		  		\STATE \textbf{End for}
    		  	\end{algorithmic}
    		  	\caption{The Gradient Projection Algorithm}
    		  \end{algorithm}
    		   \section{NUMERICAL RESULTS}
    		      		 In this section, we first verify the results of Theorem 2, then we use the gradient projection algorithm to find the optimal placement for UAV under different cases. Table I lists the parameters used in the numerical analysis.
    		      		 \begin{table}[!h]
    		      		 	\scriptsize
    		      		 	\renewcommand{\arraystretch}{1.3}
    		      		 	\caption{Parameters in numerical analysis}
    		      		 	\label{table}
    		      		 	\centering
    		      		 	\begin{tabular}{|c|c|}
    		      		 		\hline
    		      		 		Dimensions of area& $[0,250]\times[0,250]$\\
    		      		 		\hline 
    		      		 		Number of ground users & 200 users\\
    		      		 		\hline
    		      		 			Maximum number of iterations  $n_{max}$ & 100\\
    		      		 			\hline 
    		      		 			Maximum transmit power of wireless device $P_{max}$& 0.5 watt\\
    		      		 			\hline
    		      		 			Energy of each wireless device $E_i$ in joule& 4500+13500*rand(200,1)\\
    		      		 			\hline
    		      		 			Data rate $R$& 4 Mbps\\
    		      		 			\hline
    		      		 Total bandwidth $B$& 50 MHz\\
    		      		 \hline
    		      		 The noise power $N$& $1\times 10^{-14}$\\
    		      		 \hline
    		      		 The carrier frequency $f$ & 4 Ghz\\
    		      		 \hline
    		      		  Threshold time duration of uplink transmission $\tau_{th}$& 900 seconds\\
    		      		 \hline
    		      		 Minimum altitude for UAV& 650 meters\\
    		      		 \hline
    		      		 	\end{tabular}
    		      		 \end{table}
    		      		 
  To verify the results of Theorem 2, we assume that 200 ground users are uniformly distributed in a geographical area of size $250 m \times 250 m$, then we plot the objective function in (10) without any constraints at two different altitudes of the UAV. The first value for altitude $z_{min}$ is 650 meters, which is greater than $\sqrt{3} d_{max}$ and satisfies the condition in Theorem 2. The second value for altitude is 30 meters and it does not satisfy the condition in Theorem 2. In Figure 4.a, we can notice that the objective is concave when the altitude of UAV is equal to 650 meters. On the other hand, the objective function becomes non-concave at 30 meters as shown in Figure 4.b. We can also notice that the objective function at low altitude has better costs compared to results at high altitude, which make our approach practical for UAVs that have high altitude constraints like Project Loon~\cite{levy2013google} by Google .    
  
  In Figure 5, we place the UAV at 650 meters and use the gradient projection algorithm to find the optimal 2D placement that maximizes the lifetime of wireless devices when the ground users are uniformly distributed. The optimal placement for UAV is (131, 128, 650) and the optimal cost is 5.19 $J/m^2$ (282096 seconds). We can notice that the projection of the optimal point is located near the center of deployment region. This is because the devices that have minimum costs are located at the corners of the deployment region and placing the projection of UAV near the center maximizes the total cost. In Figure 6, we use the gradient projection algorithm to find the optimal 2D placement when the ground users are non-uniformly distributed. The optimal placement for UAV is (92, 156, 650) and the optimal cost is 5.22 $J/m^2$ (283727 seconds). We can notice that the placement of UAV is near the high density region.
  \begin{figure*}[!t]
  	\centering
  	\includegraphics[scale=0.3]{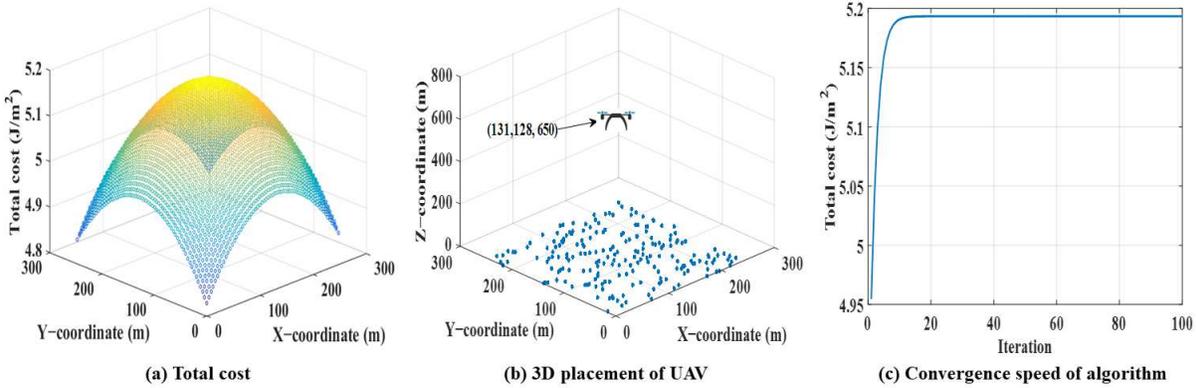}
  	\caption{Simulation results of the uniform distribution case }
  	\label{fig5}
  \end{figure*}
  \begin{figure*}[!t]
  	\centering
  	\includegraphics[scale=0.3]{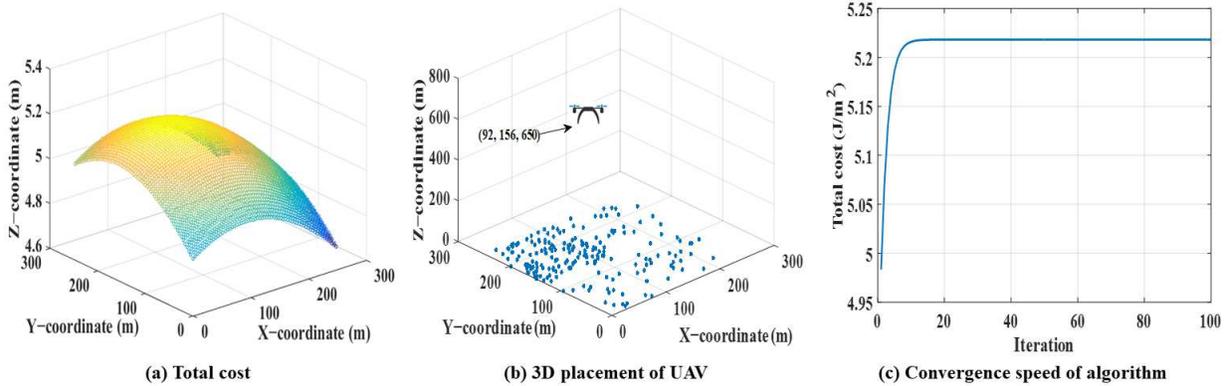}
  	\caption{Simulation results of the non-uniform distribution case }
  	\label{fig6}
  \end{figure*}
  \section{CONCLUSION} 
  In this paper, we study the problem of optimal UAV placement,
  where the objective is to determine the placement of a single UAV
  such that the sum of time durations of uplink transmissions is
  maximized. We prove that the constraint sets of problem can be
  represented by the intersection of half spheres and the region
  formed by this intersection is a convex set in terms of two
  variables. This proof enables us to transform our problem to
  an optimization problem with two variables. We also prove that
  the objective function of the transformed problem is a concave
  function under a restriction on the minimum altitude of the
  UAV and propose a gradient projection-based algorithm to find
  the optimal location of the UAV. As future work, we will study the lifetime maximization problem when multiple UAVs are utilized.
\section*{Acknowledgment}
This work was supported in part by the US NSF grants CNS-1647170 and EEC-1560131.

	\bibliographystyle{IEEEtran}
	\bibliography{UAV}

\end{document}